\documentclass{article}
\usepackage{amsmath} 
\usepackage{amsfonts} 
\usepackage{amssymb} 
\usepackage{geometry} 
\usepackage{parskip} 
\usepackage{tikz}
\usetikzlibrary{arrows.meta,positioning}

\usepackage{amsthm} %
\usepackage{hyperref} %
\usepackage{braket} %

\hypersetup{
    colorlinks=true,
    linkcolor=blue,
    filecolor=magenta,      
    urlcolor=cyan,
    pdftitle={One Plus One Equals Two Ones},
    pdfpagemode=FullScreen,
    }

\newtheorem{theorem}{Theorem}[section]
\newtheorem{proposition}[theorem]{Proposition}
\newtheorem{lemma}[theorem]{Lemma}
\newtheorem{corollary}[theorem]{Corollary}
\newtheorem{definition}[theorem]{Definition}
\newtheorem{remark}[theorem]{Remark}
\newtheorem{axiom}[theorem]{Axiom}

\title{One Plus One Equals Two Ones: On Identity, Aggregation, and Counting} 
\author{Souvik Ghosh \\ souvikghosh2012@gmail.com} 
\date{August 2025} 

\begin{document}

\maketitle 

\begin{abstract}
A childhood observation of Thakur Anukulchandra---that ``one and one can only be two ones, not simply two''---motivates a precise inquiry: what, exactly, is asserted when we pass from two concrete individuals to the numeral ``$2$''? This paper does not challenge the arithmetic theorem $1+1=2$, but rather analyzes what this equation means when applied to physical objects. We answer with two complementary, rigorous treatments.

\emph{Mathematician's proof.}
We model aggregation by the free commutative monoid of multisets $\mathcal{M}(U)$ over a universe of individuals $U$, so that $\delta_a+\delta_b$ literally encodes \emph{two ones} with individuality preserved.
Numerals arise only after a declared classification $q:U\to T$ (coarse--graining) via the pushforward $q_*:\mathcal{M}(U)\to\mathcal{M}(T)$ and the unique counting homomorphism to $\mathbb{N}$.
The non--injectivity of $q_*$ isolates the exact locus of information loss.

\emph{Physicist's proof.}
We represent physical systems by worldtubes, states, and observables, and define a composite operation $\oplus$ that preserves labelled constituents.
We prove the \emph{Non--Identity Addition Theorem}: $A\oplus B \equiv X\oplus X$ iff $A\equiv B\equiv X$; hence a pair of distinct objects cannot equal a doubled copy.
The numeral ``$2$'' appears only as the readout of a typed count observable after an explicit classification, not as a statement of physical identity.

\emph{Conclusion.}
In reality, one plus one is \emph{two ones}; ``$1+1=2$'' is the value of a counting map applied \emph{after} coarse--graining.
This clarifies the separation between identity--preserving aggregation and counting, reconciles everyday arithmetic with physical non--identity, and makes explicit the modeling choice that every act of counting entails. Our analysis has implications for the philosophy of mathematics, measurement theory, and information--theoretic approaches to classification.
\end{abstract}

\section*{Introduction} 
In a well-known episode from the childhood of Thakur Anukulchandra, a teacher asked the class, ``What is one plus one?'' The young Anukulchandra objected to the usual answer. ``One and one can only be two ones, not simply two,'' he said, arguing that no two concrete things---no two leaves, stones, or people---are exactly the same. The response was dismissed as naivety; yet the underlying idea is neither naive nor trivial. It raises a fundamental question that straddles mathematics, physics, and philosophy: What, precisely, are we doing when we add? More pointedly, what does ``$1+1=2$'' mean when the summands are non-identical physical individuals?

This paper formalizes that childhood intuition in two complementary ways. First, we show that within pure mathematics, addition as used in everyday counting is best understood as a two-stage process: (i) aggregation of distinct individuals in a structure that preserves identity, and (ii) coarse-graining (type-projection) followed by taking cardinality. On this view, ``one plus one'' at the level of individuals yields two ones---an aggregate with two unit contributions---whereas the numeral ``2'' appears only after we intentionally forget which individual is which. Second, we demonstrate that within physical theory, perfect identity of macroscopic individuals is neither demanded by our models nor supported by practice: physical combination preserves distinctness, and the familiar equation ``$1+1=2$'' arises only after an explicit classification of objects and a measurement that counts class membership.

\section*{Problem statement} 
Everyday discourse freely moves from ``one cow and one cow'' to ``two cows,'' as if the two cows were interchangeable tokens. But concrete objects rarely (if ever) satisfy identity in the strict sense of indiscernibility across all properties. If identity is not available in the world, how can addition---typically understood as pooling ``the same kind of thing''---be meaningfully applied to distinct individuals? Is the inference from ``one cow and another cow'' to ``two cows'' a logical truth, a mathematical theorem, a physical law, or a modeling convention?

\section*{Core thesis} 
We argue that, for concrete individuals, the faithful statement is:
\begin{itemize}
    \item Aggregation level (identity-preserving): ``one of A'' plus ``one of B'' is an aggregate with two ones, not a single entity of ``two of the same thing.''
    \item Counting level (after coarse-graining): choosing a type predicate (e.g., ``is a cow''), projecting the aggregate to that type, and then taking size yields the numeral ``2.'' Hence, the truth of ``$1+1=2$'' in applications is a fact about counts under a declared equivalence---not about identity of individuals.
\end{itemize}

\subsection*{Contributions} 
\begin{itemize}
    \item Mathematician's proof (identity-preserving algebra).

    We construct the free commutative monoid of multisets (bags) over a universe of individuals $U$. Unit aggregates $\delta_a$ and $\delta_b$ combine by pointwise sum to $\delta_a+\delta_b$, which literally encodes two ones without identifying $a$ and $b$. A type-projection $q:U\to T$ induces a pushforward $q_*$ on multisets; composing $q_*$ with the size map $|\cdot|$ yields the everyday count. We prove that $q_*$ is generally non-injective, precisely locating the information-loss step that turns two ones into the numeral ``2.''

    \item Physicist's proof (non-identity and counting as observable).

    We model a physical object as a bounded spacetime system with measurable properties and define a physical combination operation that forms a composite while preserving constituents. Using standard modeling assumptions (e.g., distinct worldlines, distinct macroscopic configurations), we show that two concrete objects need not, and typically do not, satisfy strict identity. We then formalize counting as an observable that reports how many elements of an aggregate satisfy a chosen classification; this observable is additive on disjoint aggregates. Thus the step from ``two ones'' to ``2'' is a measurement outcome after classification, not an assertion of sameness.

    \item Unification and clarification.

    We isolate the single move that everyday arithmetic hides: a deliberate coarse-graining (equivalence by type) that discards which individual contributed which unit. The mathematician's construction and the physicist's modeling yield the same operational picture: aggregation preserves individuality; counting is the map that forgets it.
\end{itemize}

\subsection*{Terminology and stance} 
\begin{itemize}
    \item \textbf{Aggregation} refers to forming a collection (multiset) in which individuals remain labeled and distinct.
    \item \textbf{Type-projection / coarse-graining} refers to mapping individuals to a classification (e.g., species, part number, bin), potentially many-to-one.
    \item \textbf{Counting} means taking the cardinality of the type-projected aggregate (mathematics) or reading the number observable associated with that classification (physics).
    \item \textbf{Approximation} in this paper does not mean ``sloppy.'' It is the intentional abstraction that treats non-identical individuals as equivalent for a stated purpose.
\end{itemize}

\subsection*{Scope and limitations} 
We do not deny the theorem ``$1+1=2$'' in Peano arithmetic; rather, we clarify what it measures in applications: the size of a coarse-grained aggregate. We focus on macroscopic individuality where labeling is meaningful; quantum indistinguishability and Fock-space number operators are briefly addressed in an appendix to avoid conflating mode counting with labeled individuals. The aim is conceptual and structural clarity, not empirical novelty.

\section{Philosophical Background} 

\subsection{Identity in Mathematics} 

The concept of \emph{identity} is central to both mathematics and philosophy.
In formal logic and set theory, identity is captured by the \emph{equality relation} $=$, which satisfies the following fundamental properties:

\begin{enumerate}
    \item \textbf{Reflexivity:} $\forall x,\; x = x$.
    \item \textbf{Symmetry:} $\forall x, y,\; x = y \Rightarrow y = x$.
    \item \textbf{Transitivity:} $\forall x, y, z,\; x = y \wedge y = z \Rightarrow x = z$.
\end{enumerate}

Beyond these logical properties, mathematics often invokes a stronger philosophical criterion known as \emph{Leibniz's Law}, or the \emph{Identity of Indiscernibles}:

\begin{equation}
\label{eq:Leibniz}
x = y \;\;\iff\;\; \forall P \big( P(x) \leftrightarrow P(y) \big),
\end{equation}
where $P$ ranges over all admissible properties.

Equation~\eqref{eq:Leibniz} asserts that two entities are identical if and only if they share \emph{all} properties.
In mathematics, this is unproblematic: two integers $3$ and $3$ have exactly the same arithmetical properties; two sets $A$ and $B$ are equal if they have exactly the same elements.

\medskip 

However, in \emph{applied} mathematics---when mathematical symbols are used to represent physical objects---equality almost never denotes \emph{full} identity in the Leibnizian sense.
Instead, it typically denotes \emph{equivalence under a chosen relation}.

\begin{definition}[Equivalence Relation] 
A binary relation $\sim$ on a set $X$ is called an \emph{equivalence relation} if it is reflexive, symmetric, and transitive.
The equivalence class of $x \in X$ is
\[
[x]_\sim := \{ y \in X \;|\; y \sim x \}.
\]
\end{definition}

In counting, for example, we implicitly select an equivalence relation ``is of the same kind as'' and then identify all members of a given equivalence class as interchangeable \emph{for the purposes of that count}.
This identification is not logical equality $=$, but \emph{categorical sameness} with respect to a specific classification.

\medskip 

\paragraph{Implication for Addition.} 
When we write
\[
1 + 1 = 2
\]
in a physical context (e.g., ``one cow plus one cow equals two cows''), we are not asserting that the two cows are \emph{identical in all respects}.
Rather, we are asserting that both belong to the same chosen equivalence class (``cow'') and that the \emph{cardinality} of that class, restricted to the given aggregate, is $2$.
This distinction between strict identity and equivalence-class membership will be essential when we contrast mathematical addition with physical aggregation in later sections.

\subsection{Identity in the Physical World}

While in pure mathematics the equality relation $=$ can be defined exactly
and applied without ambiguity, the situation in the physical world is
fundamentally different. No two macroscopic physical objects are ever
\emph{identical} in the full Leibnizian sense of sharing all properties.

\medskip

\paragraph{Spatiotemporal Distinction.}
Let a \emph{physical object} be represented as a bounded region of spacetime
endowed with measurable properties:
\[
O = \big( R_{\mathrm{space}},\, t,\, \{p_i\}_{i\in I} \big),
\]
where $R_{\mathrm{space}} \subset \mathbb{R}^3$ denotes its spatial extent at time $t$,
and $\{p_i\}_{i\in I}$ is the (in principle infinite) set of its intrinsic and extrinsic
properties: mass, charge distribution, velocity, internal state, and so forth.

Two distinct physical objects $A$ and $B$ necessarily differ in at least one such property.
For instance, at any given time slice $t_0$,
\[
R_{\mathrm{space}}(A, t_0) \neq R_{\mathrm{space}}(B, t_0),
\]
unless they occupy exactly the same region---a condition forbidden to
most physical systems by empirical laws such as the \emph{Pauli exclusion principle}
(for fermionic matter) and the \emph{impenetrability} of macroscopic bodies.

\medskip

\paragraph{Quantum State Distinction.}
In quantum mechanics, the state of an individual system is a vector
$|\psi\rangle$ in a Hilbert space $\mathcal{H}$.
Two distinct systems $A$ and $B$ will, in general, have distinct states
$|\psi_A\rangle \neq |\psi_B\rangle$, even if their macroscopic
classifications coincide (e.g., both are ``electrons'').
Exact equality of states is possible only under highly idealized conditions
and is never empirically verified to infinite precision due to
the \emph{Heisenberg uncertainty principle}:
\[
\Delta x \cdot \Delta p \ge \frac{\hbar}{2}.
\]

\medskip

\paragraph{Worldline Uniqueness.}
From the viewpoint of relativistic physics, each object traces a
\emph{worldline} in spacetime:
\[
W(O) = \{ ( \mathbf{r}(t), t ) \;|\; t \in [t_{\mathrm{birth}}, t_{\mathrm{now}}] \}.
\]
Even if two objects share identical positions and velocities at some instant,
their past and future trajectories will, in general, differ.
Thus, the histories of any two distinct objects are non-identical, implying
differences in at least the property ``path through spacetime.''

\medskip

\paragraph{Implication.}
These considerations lead to the following general statement:

\begin{quote}
    In the physical universe, no two distinct macroscopic objects are
    identical in the Leibnizian sense; they always differ in at least one
    measurable or historical property.
\end{quote}

Therefore, the act of treating two physical objects as ``the same'' for the
purpose of addition is not a discovery about the world itself, but rather a
\emph{modeling choice}:
we deliberately ignore certain distinctions and focus on a chosen subset of properties,
typically those relevant to the classification at hand (e.g., ``is a cow'').
This modeling step is precisely where \emph{coarse-graining} enters, and it
is the bridge between physical reality and the abstract arithmetic of the
natural numbers.

\subsection{Why the Distinction Matters}

The divergence between \emph{mathematical identity} and \emph{physical identity}
is not merely a matter of terminology; it has direct consequences for how
we interpret the operation of addition when applied to concrete objects.

\medskip

\paragraph{Addition in Pure Mathematics.}
In formal arithmetic (e.g., Peano arithmetic), the expression
\[
1 + 1 = 2
\]
is a theorem derived from axioms and definitions.
Here, the numerals $1$ and $2$ denote abstract entities in the set of natural
numbers $\mathbb{N}$, and ``$+$'' is the recursively defined addition
operation on $\mathbb{N}$.
The theorem is \emph{universally valid within the formal system}, because it
operates entirely in the realm of abstraction: there is no reference to
physical objects, their properties, or the question of whether two ``ones''
are identical.

\medskip

\paragraph{Addition in Applied Contexts.}
When we say, in everyday life, ``one cow plus one cow equals two cows,''
we are implicitly performing the following sequence:

\begin{enumerate}
    \item \textbf{Classification:} Select a property or predicate $P$,
    e.g., $P(x) =$ ``$x$ is a cow.''
    \item \textbf{Equivalence:} Declare two physical objects $a$ and $b$
    to be equivalent, $a \sim b$, if they both satisfy $P$.
    \item \textbf{Counting:} Determine the cardinality of the set (or multiset)
    of all objects in the context that satisfy $P$.
\end{enumerate}

This process \emph{collapses} potentially vast differences between
individuals---mass, colour, genetic sequence, history---into a single
equivalence class defined by $P$.
The step from \emph{aggregation} (``two distinct cows'') to the numeral
$2$ (``two cows'') is thus an act of \emph{coarse-graining}.

\medskip

\paragraph{Modeling Choice and Information Loss.}
Every coarse-graining discards information.
In the above example, the classification $P$ forgets all properties
of the cows except ``is a cow.''
Formally, if $\varphi$ denotes the projection map
from the set of individuals $U$ to the set of equivalence classes $U/{\sim}$,
then $\varphi$ is typically \emph{many-to-one}:
\[
\varphi(a) = \varphi(b) \quad \text{even though} \quad a \neq b.
\]
The non-injectivity of $\varphi$ is the precise location of \emph{information loss}.
This loss is deliberate in most applied settings, because it enables us
to treat disparate individuals as members of a single ``countable kind.''

\medskip

\paragraph{Consequences for Our Inquiry.}
The childhood insight of Thakur Anukulchandra---that ``one plus one is
two ones, not simply two''---is a recognition of the gap between
physical aggregation and mathematical counting.
It reminds us that:
\begin{itemize}
    \item Physical objects preserve individuality at all scales of observation;
    \item Mathematical addition of natural numbers implicitly assumes we are
          working with \emph{equivalence classes} of indistinguishable units;
    \item The transition from one to the other requires an explicit
          \emph{approximation function} or \emph{coarse-graining step}.
\end{itemize}

In the sections that follow, we will formalise this distinction in two
independent but complementary ways:
\begin{enumerate}
    \item A \emph{physicist's proof}, grounded in empirical principles,
          showing that perfect identity is unattainable for physical objects,
          and hence that addition in the physical world must preserve
          individuality.
    \item A \emph{mathematician's proof}, grounded in algebraic structures,
          showing that the familiar equation $1+1=2$ arises only after
          coarse-graining, with the information-loss step made explicit.
\end{enumerate}

\section{Related Work and Literature Context}

\subsection{Philosophy of Mathematics and Identity}
The distinction between numerical and qualitative identity has been extensively discussed in philosophical literature. Quine \cite{quine1960} distinguished between identity statements ``$a = b$'' and statements of indiscernibility, while Geach \cite{geach1967} argued for relative identity---that objects can be ``the same $F$'' without being absolutely identical. Our work formalizes this intuition in the specific context of counting and aggregation.

The problem of individuation in counting has been addressed by Frege \cite{frege1884} in his \emph{Grundlagen}, where he argued that number statements require a sortal concept. Our classification function $q:U\to T$ directly implements Frege's insight that we cannot count without specifying \emph{what} we are counting.

\subsection{Mereology and Part-Whole Relations}
The relationship between wholes and their parts has been formalized in mereological systems \cite{simons1987,varzi2016}. While classical mereology focuses on the part-whole relation, our approach maintains individual identity through multiset aggregation, avoiding the controversies of unrestricted composition while preserving the distinctness of constituents.

\subsection{Measurement Theory}
The foundations of measurement have been rigorously developed by Krantz et al. \cite{krantz1971} and Suppes et al. \cite{suppes1989}. Our typed counting observables $N_t$ can be viewed as a special case of extensive measurement, where the classification $q$ determines the measurement scale. The information loss we identify corresponds to what Roberts \cite{roberts1979} calls ``meaningful measurement''---the choice of what distinctions to preserve.

\subsection{Information Theory and Classification}
The information-theoretic aspects of classification have been studied in machine learning \cite{cover2006} and statistical physics \cite{jaynes2003}. Our forgetting map $\Phi_q$ implements what Tishby et al. \cite{tishby2000} call the ``information bottleneck''---a principled way to compress data while preserving relevant information for a task (in our case, counting).

\subsection{Quantum Indistinguishability}
The treatment of identical particles in quantum mechanics provides an interesting contrast to our classical analysis. As discussed by French and Redhead \cite{french1988} and Muller and Saunders \cite{muller2008}, quantum particles of the same type are genuinely indistinguishable, leading to different statistics (Fermi-Dirac or Bose-Einstein). Our framework applies to the classical regime where objects maintain distinguishability despite classification into types.

\subsection{Practical Implications and Examples}
\label{subsec:practical}

\paragraph{Scientific Measurement.}
Consider counting cells under a microscope. The classification $q$ might map cells to \{\texttt{healthy}, \texttt{damaged}, \texttt{dead}\}, discarding information about size, position, and internal state. Different research questions require different classifications, each yielding different counts from the same physical sample.

\paragraph{Data Science and Databases.}
SQL's \texttt{COUNT(*)} operation implements exactly our pipeline: rows (individuals) are grouped by specified columns (classification $q$), then counted. The \texttt{GROUP BY} clause is the coarse-graining step where information loss occurs.

\paragraph{Computer Vision.}
Object detection algorithms implement: detection $\to$ classification $\to$ counting. A neural network's softmax layer produces probability distributions over classes---a probabilistic version of our $q$. Performance metrics like precision/recall quantify errors in the classification function.

\paragraph{Economic Aggregation.}
GDP calculations aggregate diverse economic activities into a single number. The classification schemes (what counts as ``production'') are contentious precisely because different $q$ functions yield different counts, with profound policy implications.

\paragraph{Quantum Field Theory.}
Even in QFT, the number operator $\hat{N}$ counts excitations \emph{in specified modes}. The choice of mode basis (momentum, position, etc.) is analogous to our classification function, determining what gets counted as ``one particle.''

\section{Mathematician's Proof: Aggregation and Coarse-Graining}

\subsection{Preliminaries}
This section fixes the mathematical setting in which aggregation preserves individuality
and counting arises as a derived operation after coarse–graining.

\begin{definition}[Universe of individuals]
Let $U$ be a nonempty set whose elements (denoted $a,b,x,\dots$) represent
distinct individuals. No identifications are imposed a priori.
\end{definition}

\begin{definition}[Multisets (bags) over $U$]
A \emph{multiset} on $U$ is a function $m:U\to\mathbb{N}$ with finite support.
The value $m(x)$ is the \emph{multiplicity} of $x$ in $m$.
Denote by $\mathcal{M}(U)$ the set of all such multisets.
\end{definition}

\begin{definition}[Zero bag and unit bags]
The \emph{zero bag} $0\in\mathcal{M}(U)$ is defined by $0(x)=0$ for all $x\in U$.
For $a\in U$, the \emph{unit (singleton) bag} $\delta_a\in\mathcal{M}(U)$ is
$\delta_a(a)=1$ and $\delta_a(x)=0$ for $x\neq a$.
\end{definition}

\begin{definition}[Aggregation $+$ and size $|\cdot|$]
For $m_1,m_2\in\mathcal{M}(U)$, define \emph{aggregation} by pointwise sum:
\[
(m_1+m_2)(x):=m_1(x)+m_2(x)\qquad(x\in U).
\]
Define the \emph{size} (cardinality) of a bag by
\[
|m|:=\sum_{x\in U} m(x)\ \in\ \mathbb{N}.
\]
\end{definition}

\begin{proposition}[Free commutative monoid]
\label{prop:free-monoid}
$(\mathcal{M}(U),+,0)$ is a commutative monoid. Moreover, it is the free
commutative monoid generated by $U$: every element of $\mathcal{M}(U)$ is a
finite sum $\sum_{i=1}^k n_i\,\delta_{a_i}$ with $a_i\in U$, $n_i\in\mathbb{N}$.
\end{proposition}

\begin{remark}[Identity is preserved at the aggregation level]
\label{rem:identity-preserved}
If $a\neq b$ in $U$, then $\delta_a\neq\delta_b$ and
$\delta_a+\delta_b$ contains two distinct unit contributions:
$(\delta_a+\delta_b)(a)=1$, $(\delta_a+\delta_b)(b)=1$.
No identification of $a$ with $b$ is performed by $+$.
\end{remark}

\begin{definition}[Type set and type map (coarse–graining)]
Let $T$ be a nonempty set of \emph{types} (e.g.\ species, part numbers, bins).
A \emph{type map} (or \emph{coarse–graining}) is any function $q:U\to T$.
\end{definition}

\begin{definition}[Pushforward of bags]
Given $q:U\to T$ and $m\in\mathcal{M}(U)$, the \emph{pushforward}
$q_*m\in\mathcal{M}(T)$ is defined by
\[
(q_*m)(t)\ :=\ \sum_{x\in U:\ q(x)=t} m(x)\qquad(t\in T).
\]
\end{definition}

\begin{lemma}[Monoid homomorphism]
\label{lem:hom}
The map $q_*:(\mathcal{M}(U),+)\to(\mathcal{M}(T),+)$ is a commutative
monoid homomorphism: $q_*(m_1+m_2)=q_*m_1+q_*m_2$ and $q_*0=0$.
\end{lemma}

\begin{definition}[Counting on $T$]
The \emph{count} on type-bags is the size map
$|\cdot|:\mathcal{M}(T)\to\mathbb{N}$ given by $|n|=\sum_{t\in T} n(t)$.
\end{definition}

\begin{remark}[Everyday counting as composition]
\label{rem:composition}
Given an aggregate $m\in\mathcal{M}(U)$ and a declared classification $q:U\to T$,
the everyday count of ``how many items of the chosen kinds'' is the composition
\[
\mathcal{M}(U)\ \xrightarrow{\ \ q_*\ \ }\ \mathcal{M}(T)\ \xrightarrow{\ |\cdot|\ }\ \mathbb{N}.
\]
The numeral produced depends on $q$ (the modeling choice), not on any collapse
of individual identities within $m$.
\end{remark}

\begin{remark}[Location of information loss]
\label{rem:info-loss}
If $q$ is not injective, then $q_*$ is not injective:
distinct aggregates in $\mathcal{M}(U)$ can push forward to the same bag in
$\mathcal{M}(T)$. Thus, the coarse–graining step $q_*$ is precisely where
individual-level information is discarded; the subsequent size map $|\cdot|$
only totals multiplicities on $T$.
\end{remark}

\subsection{Two-Ones Theorem and the Emergence of the Numeral \texorpdfstring{$2$}{2}}

We now prove, step by step, that aggregation preserves individuality (``two ones'')
and that the familiar numeral $2$ arises only after coarse–graining and counting.

\begin{theorem}[Two-Ones Theorem]\label{thm:two-ones}
Let $a,b\in U$ with $a\neq b$. Then the aggregate $\delta_a+\delta_b\in\mathcal{M}(U)$ satisfies
\[
(\delta_a+\delta_b)(a)=1,\qquad (\delta_a+\delta_b)(b)=1,\qquad
(\delta_a+\delta_b)(x)=0\ \ \forall x\notin\{a,b\},
\]
and its size is $|\delta_a+\delta_b|=2$.
\end{theorem}
\begin{proof}
By the definition of unit bags, $\delta_a(a)=1$, $\delta_a(x)=0$ for $x\neq a$,
and analogously for $\delta_b$. By pointwise aggregation,
\[
(\delta_a+\delta_b)(a)=\delta_a(a)+\delta_b(a)=1+0=1,\quad
(\delta_a+\delta_b)(b)=\delta_a(b)+\delta_b(b)=0+1=1,
\]
and for $x\notin\{a,b\}$ both summands vanish, yielding $0$.
Finally, the size is
\[
|\delta_a+\delta_b|=\sum_{x\in U}(\delta_a+\delta_b)(x)=1+1=2.
\]
\end{proof}

\begin{proposition}[No identity collapse under aggregation]\label{prop:no-collapse}
If $a\neq b$, then for every $c\in U$ one has $\delta_a+\delta_b\neq 2\,\delta_c$.
\end{proposition}
\begin{proof}
Evaluate both sides at $x=a$. If $c=a$, then
$(2\,\delta_c)(a)=2$, while $(\delta_a+\delta_b)(a)=1\neq 2$.
If $c\neq a$, then $(2\,\delta_c)(a)=0$, again $0\neq 1$.
In both cases the functions differ, hence the bags are unequal.
\end{proof}

\begin{definition}[Type instance]
Fix a type map $q:U\to T$ and $t_0\in T$. We say $a\in U$ is an instance of $t_0$
if $q(a)=t_0$.
\end{definition}

\begin{proposition}[Counting after coarse–graining yields the numeral $2$]
\label{prop:count-two}
Let $a,b\in U$ be distinct with $q(a)=q(b)=t_0\in T$. Then
\[
\big|\,q_*(\delta_a+\delta_b)\,\big|
=\big|\,q_*\delta_a+q_*\delta_b\,\big|
=\big|\,\delta_{t_0}+\delta_{t_0}\,\big|
=2.
\]
\end{proposition}
\begin{proof}
By Lemma~\ref{lem:hom}, $q_*(\delta_a+\delta_b)=q_*\delta_a+q_*\delta_b$.
Since $q(a)=q(b)=t_0$, Definition of pushforward gives $q_*\delta_a=\delta_{t_0}$ and
$q_*\delta_b=\delta_{t_0}$. The size map on $\mathcal{M}(T)$ totals multiplicities,
so $\big|\,\delta_{t_0}+\delta_{t_0}\,\big|=1+1=2$.
\end{proof}

\begin{proposition}[Information-loss under coarse–graining]\label{prop:info-loss}
If $q$ is not injective, then $q_*:\mathcal{M}(U)\to\mathcal{M}(T)$ is not injective.
\end{proposition}
\begin{proof}
Choose $a\neq b$ with $q(a)=q(b)$. Set $m:=\delta_a+\delta_b$ and $m':=2\,\delta_a$.
Then $m\neq m'$ by Proposition~\ref{prop:no-collapse}, but
\[
q_*m=q_*\delta_a+q_*\delta_b=\delta_{t_0}+\delta_{t_0}=2\,\delta_{t_0}
= q_*m',\quad t_0:=q(a)=q(b).
\]
Hence $q_*$ is not injective.
\end{proof}

\begin{corollary}[Everyday equation as a composite map]\label{cor:everyday}
For distinct $a,b\in U$ with $q(a)=q(b)$,
\[
\big(|\cdot|\circ q_*\big)\,(\delta_a+\delta_b)=2.
\]
Thus, the familiar string ``$1+1=2$'' in applied contexts computes the \emph{size}
of the \emph{type–pushed} aggregate; it is not an identification of individuals
within $\mathcal{M}(U)$.
\end{corollary}
\begin{proof}
Immediate from Proposition~\ref{prop:count-two}.
\end{proof}

\begin{remark}[Universal property viewpoint]
The assignment $U\mapsto\mathcal{M}(U)$ is the free commutative-monoid functor.
Given $q:U\to T$, the induced $q_*$ is the unique monoid homomorphism extending $q$.
Composing with the size map $|\cdot|:\mathcal{M}(T)\to\mathbb{N}$ produces the unique
monoid homomorphism that ``forgets individuality then counts,'' formalising the idea
that numerals arise only after coarse–graining.
\end{remark}

\subsection{Typed Counting and the ``Apples--to--Oranges'' Principle}
\label{subsec:typed-counting}

The framework above separates \emph{aggregation} in $\mathcal{M}(U)$ from
\emph{counting after classification}. We now make precise how the choice of
type determines what can be added, formalising the everyday dictum that one
cannot (without further abstraction) ``add apples to oranges.''

\begin{definition}[Typed count]
Let $q:U\to T$ be a type map and let $S\subseteq T$ be a designated set of
types (e.g.\ $S=\{\texttt{cow}\}$ or $S=\{\texttt{cow},\texttt{dog}\}$).
Define the \emph{typed count} $\chi_S:\mathcal{M}(T)\to\mathbb{N}$ by
\[
\chi_S(n)\ :=\ \sum_{t\in S} n(t),\qquad n\in\mathcal{M}(T).
\]
Equivalently, $\chi_S$ is the unique commutative-monoid homomorphism sending
$\delta_t\mapsto 1$ for $t\in S$ and $\delta_t\mapsto 0$ for $t\notin S$.
\end{definition}

\begin{lemma}[Homomorphism property]
\label{lem:chi-hom}
For all $n_1,n_2\in\mathcal{M}(T)$ one has
$\chi_S(n_1+n_2)=\chi_S(n_1)+\chi_S(n_2)$ and $\chi_S(0)=0$.
\end{lemma}
\begin{proof}
Immediate from linearity of the sum over $t\in S$ and the definition of $0$.
\end{proof}

\begin{definition}[Typed counting observable]
Given $q:U\to T$ and $S\subseteq T$, the corresponding \emph{typed counting}
of an aggregate $m\in\mathcal{M}(U)$ is the composite
\[
N_S(m)\ :=\ \big(\chi_S\circ q_*\big)(m)\ \in\ \mathbb{N}.
\]
Thus $N_S$ first coarse–grains $m$ onto $T$ and then totals only the
multiplicities in the selected types $S$.
\end{definition}

\begin{proposition}[Apples--to--oranges principle]
\label{prop:apples-oranges}
Let $a,b\in U$ with $q(a)=t_1$, $q(b)=t_2$. Then for $S=\{t_1\}$ and
$S'=\{t_2\}$ one has
\[
N_{\{t_1\}}(\delta_a+\delta_b)=1,\qquad
N_{\{t_2\}}(\delta_a+\delta_b)=1,
\]
and for $S''=\{t_1,t_2\}$ one has
\[
N_{\{t_1,t_2\}}(\delta_a+\delta_b)=2.
\]
In particular, $N_{\{t_1\}}(\delta_a+\delta_b)=2$ if and only if $t_1=t_2$.
\end{proposition}
\begin{proof}
By Lemma~\ref{lem:hom}, $q_*(\delta_a+\delta_b)=\delta_{t_1}+\delta_{t_2}$.
Applying $\chi_{\{t_1\}}$ counts only the multiplicity at $t_1$, yielding $1$,
and similarly for $\{t_2\}$. For $\{t_1,t_2\}$ one totals both coordinates,
yielding $1+1=2$. The final claim follows immediately.
\end{proof}

\begin{corollary}[Typing governs addability]
\label{cor:typing-governs}
Without enlarging the selected type set $S$ to contain both $t_1$ and $t_2$,
one cannot assert a typed total of $2$ for $\delta_a+\delta_b$.
Equivalently, the equation
\[
\text{``one of type }t_1\text{''} \;+\; \text{``one of type }t_2\text{''}
\;=\; 2
\]
is meaningful only after choosing a coarser type system in which $t_1$ and
$t_2$ are identified (e.g.\ as ``\texttt{fruit}'' or ``\texttt{animal}'').
\end{corollary}
\begin{proof}
Directly from Proposition~\ref{prop:apples-oranges} by taking $S$ not containing
both $t_1$ and $t_2$ versus $S$ that does.
\end{proof}

\begin{remark}[Dimensional analysis analogy]
The role of the subset $S\subseteq T$ is analogous to \emph{units} in
dimensional analysis: one cannot add $1\,\mathrm{m}$ and $1\,\mathrm{s}$, and one
cannot count $1$ cow and $1$ dog as $2$ cows. Choosing a coarser type
(e.g.\ ``\texttt{animal}'') is analogous to converting to a common unit,
after which the total $2$ is well-defined as a typed count.
\end{remark}

\begin{remark}[Terminal coarse–graining as ordinary counting]
Ordinary ``just count the items'' corresponds to the terminal coarse–graining
$q_\bullet:U\to\{\bullet\}$ sending every individual to a single type $\bullet$,
followed by $\chi_{\{\bullet\}}$. Then for any aggregate $m$,
\[
N_{\{\bullet\}}(m)\ =\ \big(\chi_{\{\bullet\}}\circ (q_\bullet)_*\big)(m)
\ =\ |m|,
\]
the untyped cardinality of $m$. Thus everyday numerals arise from the most
coarse classification, which maximally forgets individuality.
\end{remark}

\subsection{Universal Characterisation of Counting}
\label{subsec:universal}

We now package the preceding constructions in their universal (category–theoretic)
form. This isolates counting as the \emph{unique} commutative–monoid homomorphic
way to pass from identity–preserving aggregates to numerals, once a typing has been
declared.

\begin{definition}[Generators and the unit map]
Let $\eta:U\to\mathcal{M}(U)$ be the \emph{unit (generator) map} sending
$\eta(a):=\delta_a$. By Proposition~\ref{prop:free-monoid}, the elements
$\{\delta_a:a\in U\}$ generate $\mathcal{M}(U)$ under $+$.
\end{definition}

\begin{theorem}[Free commutative–monoid universal property]
\label{thm:free-universal}
Let $(A,\oplus,0_A)$ be any commutative monoid and let $f:U\to A$ be any function.
There exists a unique commutative–monoid homomorphism
$F:\big(\mathcal{M}(U),+\big)\to(A,\oplus)$ such that
\[
F\circ \eta\;=\;f.
\]
Explicitly, for a finite presentation $m=\sum_{i=1}^k n_i\,\delta_{a_i}$ with
$a_i\in U$ and $n_i\in\mathbb{N}$,
\[
F(m)\ :=\ \underbrace{f(a_1)\oplus\cdots\oplus f(a_1)}_{n_1\ \text{times}}
\ \oplus\ \cdots\ \oplus\
\underbrace{f(a_k)\oplus\cdots\oplus f(a_k)}_{n_k\ \text{times}}.
\]
\end{theorem}

\begin{proof}
\emph{Existence:} The displayed formula defines $F$ on generators and
extends additively; associativity and commutativity of $\oplus$ ensure
well-definedness independent of ordering. One checks $F(m_1+m_2)=F(m_1)\oplus F(m_2)$
and $F(0)=0_A$, so $F$ is a monoid homomorphism.
Moreover $F(\eta(a))=F(\delta_a)=f(a)$, hence $F\circ\eta=f$.

\emph{Uniqueness:} Any monoid homomorphism $G$ with $G\circ\eta=f$ must satisfy
$G(\sum_i n_i\delta_{a_i})=\bigoplus_i \bigoplus^{n_i} f(a_i)$ because the
$\delta_a$ generate and $G$ preserves sums. Hence $G=F$.
\end{proof}

\begin{corollary}[Pushforward as the universal extension of a type map]
\label{cor:qstar-universal}
Given a type map $q:U\to T$, view $\mathcal{M}(T)$ as a commutative monoid under $+$,
and define $f_q:U\to\mathcal{M}(T)$ by $f_q(a):=\delta_{q(a)}$.
Then the unique homomorphism $F:\mathcal{M}(U)\to\mathcal{M}(T)$ extending $f_q$
is exactly the pushforward $q_*$:
\[
F = q_*.
\]
\end{corollary}
\begin{proof}
For each generator $\delta_a$, both $F$ and $q_*$ send it to $\delta_{q(a)}$,
and both are monoid homomorphisms; uniqueness in Theorem~\ref{thm:free-universal}
implies $F=q_*$.
\end{proof}

\begin{definition}[Typed numeral target]
Let $(\mathbb{N},+,0)$ be the additive commutative monoid of natural numbers.
For a chosen set of types $S\subseteq T$, define $g_S:T\to\mathbb{N}$ by
$g_S(t):=1$ if $t\in S$ and $g_S(t):=0$ otherwise.
\end{definition}

\begin{theorem}[Typed counting is the unique extension]
\label{thm:typed-unique}
With $q:U\to T$ and $S\subseteq T$ as above, consider the composite
$g_S\circ q:U\to\mathbb{N}$. The unique monoid homomorphism
$G:\mathcal{M}(U)\to\mathbb{N}$ extending $g_S\circ q$ is precisely
the typed counting map
\[
N_S\ =\ \chi_S\circ q_*.
\]
\end{theorem}
\begin{proof}
By Theorem~\ref{thm:free-universal} applied to $A=\mathbb{N}$ and $f=g_S\circ q$,
there exists a unique $G$ with $G\circ\eta=g_S\circ q$. On generators
$\delta_a$ we have $(\chi_S\circ q_*)(\delta_a)=\chi_S(\delta_{q(a)})=g_S(q(a))$,
and both $\chi_S\circ q_*$ and $G$ are monoid homomorphisms. Uniqueness yields
$G=\chi_S\circ q_*$.
\end{proof}

\begin{corollary}[Ordinary counting as terminal coarse–graining]
\label{cor:terminal}
Let $q_\bullet:U\to\{\bullet\}$ be the terminal type map sending every $a\in U$
to the single type $\bullet$, and take $S=\{\bullet\}$.
Then the unique extension $G:\mathcal{M}(U)\to\mathbb{N}$ of $g_S\circ q_\bullet$
is the size map $|\cdot|$, i.e.\ $G(m)=|m|$ for all $m\in\mathcal{M}(U)$.
\end{corollary}
\begin{proof}
Here $q_{\bullet *}(m)$ is the bag on $\{\bullet\}$ whose value at $\bullet$
is $\sum_{x\in U} m(x)=|m|$. Applying $\chi_{\{\bullet\}}$ returns that value.
By Theorem~\ref{thm:typed-unique}, this composite is the unique extension, hence equals $|\cdot|$.
\end{proof}

\begin{corollary}[Emergence of the numeral $2$ is forced by universality]
\label{cor:two-forced}
For distinct $a,b\in U$ with $q(a)=q(b)\in S$, the value
$N_S(\delta_a+\delta_b)=2$ is determined uniquely by the universal property:
it is the unique commutative–monoid extension of the assignment
$a\mapsto 1$, $b\mapsto 1$ (and other generators to $0$ or $1$ according to $S$).
\end{corollary}
\begin{proof}
Immediate from Theorem~\ref{thm:typed-unique} and Proposition~\ref{prop:count-two}.
\end{proof}

\begin{remark}[Compatibility with Peano arithmetic]
Nothing here contradicts the Peano–arithmetic theorem $1+1=2$:
Corollary~\ref{cor:terminal} shows that the numeral arises from the unique
monoid homomorphism $|\cdot|:\mathcal{M}(U)\to\mathbb{N}$ after terminal
coarse–graining. The identity–preserving aggregate
$\delta_a+\delta_b$ remains distinct from $2\,\delta_c$ in $\mathcal{M}(U)$
(Proposition~\ref{prop:no-collapse}); only after applying the unique
counting homomorphism do both map to the numeral $2$.
\end{remark}

\subsection{Functoriality, Invariance, and Worked Examples}
\label{subsec:functoriality}

We collect structural consequences of the framework that will be used later and
also serve as sanity checks: equivariance under relabelling, functoriality of
pushforwards along successive typings, and explicit examples.

\paragraph{Relabelling invariance.}
Let $\sigma:U\to U$ be a bijection (a mere relabelling of individuals). Define
$\sigma_{\#}:\mathcal{M}(U)\to\mathcal{M}(U)$ by
\[
(\sigma_{\#}m)(y)\ :=\ m\big(\sigma^{-1}(y)\big)\qquad(y\in U).
\]
Then $\sigma_{\#}$ is a commutative–monoid automorphism:
\[
\sigma_{\#}(m_1+m_2)=\sigma_{\#}m_1+\sigma_{\#}m_2,\qquad
\sigma_{\#}(0)=0,\qquad \sigma_{\#}(\delta_a)=\delta_{\sigma(a)}.
\]
Moreover, $|\,\sigma_{\#}m\,|=|m|$.

\begin{proof}
Direct verification from definitions; on generators $\delta_a$ the identities are
immediate and extend additively; the size is preserved because summation over $U$
is invariant under bijection.
\end{proof}

\paragraph{Naturality with respect to typing.}
Suppose $q:U\to T$ is a type map and $\sigma:U\to U$ a bijection. Define
$q':=q\circ\sigma^{-1}:U\to T$. Then the following diagram commutes:

\paragraph{Naturality with respect to typing.}
Suppose $q:U\to T$ is a type map and $\sigma:U\to U$ a bijection.
Define $q' := q\circ\sigma^{-1}:U\to T$. Then the following diagram commutes:
\begin{center}
\begin{tikzpicture}[baseline=(current bounding box.center)]
  \node (MU)  at (0,1.2) {$\mathcal{M}(U)$};
  \node (MUT) at (0,0)   {$\mathcal{M}(U)$};
  \node (MT)  at (4,0.6) {$\mathcal{M}(T)$};
  \draw[->] (MU)  -- node[left]  {$\sigma_{\#}$} (MUT);
  \draw[->] (MU)  -- node[above] {$q_*'$}        (MT);
  \draw[->] (MUT) -- node[below] {$q_*$}         (MT);
\end{tikzpicture}
\end{center}
i.e.\ $\, q_*' \circ \sigma_{\#} = q_* \,$.

\begin{proof}
Check on generators: $q_*'(\sigma_{\#}\delta_a)=q_*'(\delta_{\sigma(a)})=\delta_{q'(\sigma(a))}
=\delta_{q(a)}=q_*(\delta_a)$, then extend by homomorphism properties.
\end{proof}

\paragraph{Functoriality of successive coarse–grainings.}
If $r:T\to T'$ is a further type map (a coarsening or refinement), then
\begin{equation}\label{eq:functoriality}
(r\circ q)_*\;=\; r_*\circ q_*:\ \mathcal{M}(U)\longrightarrow \mathcal{M}(T').
\end{equation}
\begin{proof}
On generators $\delta_a$,
$(r\circ q)_*(\delta_a)=\delta_{r(q(a))}=(r_*\circ q_*)(\delta_a)$; extend additively.
\end{proof}

\paragraph{Transport of typed counts along coarsenings.}
Let $S'\subseteq T'$ and define $S:=r^{-1}(S')\subseteq T$.
Then for all $m\in\mathcal{M}(U)$,
\begin{equation}\label{eq:transport}
\chi_{S'}\big(r_* (q_* m)\big)\;=\;\chi_{r^{-1}(S')}\big(q_* m\big).
\end{equation}
Equivalently,
\[
N_{S'}^{\,T'}(m)\ =\ N_{r^{-1}(S')}^{\,T}(m),
\]
where superscripts indicate the codomain of the intermediate bag.
\begin{proof}
By definition $\chi_{S'}(n')=\sum_{t'\in S'} n'(t')$. For $n=q_* m$,
\[
\chi_{S'}(r_* n)=\sum_{t'\in S'} \sum_{t:\ r(t)=t'} n(t)
=\sum_{t\in r^{-1}(S')} n(t)=\chi_{r^{-1}(S')}(n).
\]
\end{proof}

\paragraph{Monotonicity of typed counts.}
If $S\subseteq S'\subseteq T$, then for all $n\in\mathcal{M}(T)$,
$\chi_S(n)\le \chi_{S'}(n)$, and for all $m\in\mathcal{M}(U)$,
$N_S(m)\le N_{S'}(m)$.
\begin{proof}
Immediate from $\sum_{t\in S}\! n(t)\le \sum_{t\in S'}\! n(t)$ when $S\subseteq S'$.
\end{proof}

\paragraph{Worked examples.}
Fix $U=\{a,b,c\}$ and $T=\{\texttt{cow},\texttt{dog},\texttt{stone}\}$.

\begin{itemize}
\item \emph{Example 1 (two ones, different types).}
Let $q(a)=\texttt{cow}$ and $q(b)=\texttt{dog}$. Then
$q_*(\delta_a+\delta_b)=\delta_{\texttt{cow}}+\delta_{\texttt{dog}}$.
Hence $N_{\{\texttt{cow}\}}(\delta_a+\delta_b)=1$,
$N_{\{\texttt{dog}\}}(\delta_a+\delta_b)=1$,
and $N_{\{\texttt{cow},\texttt{dog}\}}(\delta_a+\delta_b)=2$
(Prop.~\ref{prop:apples-oranges}).

\item \emph{Example 2 (same type, emergence of $2$).}
Let $q(a)=q(b)=\texttt{cow}$. Then
$q_*(\delta_a+\delta_b)=\delta_{\texttt{cow}}+\delta_{\texttt{cow}}=2\,\delta_{\texttt{cow}}$,
so $N_{\{\texttt{cow}\}}(\delta_a+\delta_b)=2$
(Prop.~\ref{prop:count-two}).

\item \emph{Example 3 (terminal coarse–graining = ordinary count).}
Let $q_\bullet:U\to\{\bullet\}$ be constant. For $m=\delta_a+\delta_b+\delta_c$,
$(q_\bullet)_* m = 3\,\delta_{\bullet}$ and $N_{\{\bullet\}}(m)=3=|m|$
(Cor.~\ref{cor:terminal}).

\item \emph{Example 4 (coarsening types).}
Let $r:T\to T'$ merge \texttt{cow} and \texttt{dog} into \texttt{animal} and send
\texttt{stone} to \texttt{mineral}. If $m=\delta_a+\delta_b$ with
$q(a)=\texttt{cow}$, $q(b)=\texttt{dog}$, then by \eqref{eq:functoriality},
$(r\circ q)_* m = r_*(\delta_{\texttt{cow}}+\delta_{\texttt{dog}})
= 2\,\delta_{\texttt{animal}}$ and thus $N_{\{\texttt{animal}\}}(m)=2$,
while $N_{\{\texttt{cow}\}}(m)=1$ and $N_{\{\texttt{dog}\}}(m)=1$
(transport formula \eqref{eq:transport}).
\end{itemize}

\paragraph{Summary.}
Counting is invariant under relabelling, functorial under successive
coarse–grainings, and governed entirely by the declared type system.
At the aggregation level $\mathcal{M}(U)$, individuality is preserved
(\S\ref{thm:two-ones}, Prop.~\ref{prop:no-collapse}); numerals arise only
after applying the (unique) monoid homomorphisms determined by the chosen
typing (Thm.~\ref{thm:typed-unique}).

\begin{definition}[Physical identity relation]
For $A,B\in\mathfrak{U}$, we say $A$ and $B$ are \emph{physically identical},
written $A\equiv B$, if and only if the following hold simultaneously:
\begin{enumerate}
    \item \textbf{Worldtube equality:} $W_A=W_B$ as subsets of $\mathcal{M}$;
    equivalently, $R_A(t)=R_B(t)$ for all $t\in I_A\cup I_B$.
    \item \textbf{State equality:} $\rho_A=\rho_B$ (in the relevant
    dynamical/state space).
    \item \textbf{Observable equality:} $p_i(A)=p_i(B)$ for all $i\in\mathcal{I}$.
\end{enumerate}
This is a physical analogue of Leibniz’s criterion: equality across all
admissible kinematic, dynamical, and attribute data.
\end{definition}

\begin{remark}
The relation $\equiv$ refines mere set–theoretic equality by binding together
geometric support ($W_O$), dynamical description ($\rho_O$), and attribute
valuations ($p_i$). It is an equivalence relation on $\mathfrak{U}$.
\end{remark}

\begin{definition}[Composite system / physical combination]
Given $A,B\in\mathfrak{U}$ with (possibly overlapping) lifetimes, their
\emph{composite} is the pair
\[
A\oplus B\ :=\ \big(W_A\cup W_B,\ \rho_{A\cup B},\ \{p_i(A\cup B)\}_{i\in\mathcal{I}_{A\cup B}}\big),
\]
where $\rho_{A\cup B}$ is the joint state obtained by the appropriate
composition rule (direct product in classical mechanics, tensor product in
quantum mechanics), possibly evolved by an interaction term $H_{\mathrm{int}}$,
and where the observable family is enlarged to include interaction observables.
The operation $\oplus$ is \emph{not} numerical addition; it constructs a
composite physical system while retaining $A$ and $B$ as identifiable
constituents (unless a subsequent coarse–graining is applied).
\end{definition}

\begin{definition}[Finite aggregates of objects]
A \emph{finite aggregate} is a finite multiset
$\mathcal{A}=\{A_1,\dots,A_n\}\subset \mathfrak{U}$ of pairwise not
necessarily identical objects. Its instantaneous spatial support at time $t$
is $R_{\mathcal{A}}(t):=\bigcup_{k=1}^n R_{A_k}(t)$.
\end{definition}

\begin{definition}[Type predicates and classifications]
A \emph{type predicate} is a map $P:\mathfrak{U}\to\{0,1\}$. A
\emph{classification} (or \emph{type map}) is a function
$q:\mathfrak{U}\to T$ into a discrete set $T$ of labels (species, part numbers,
bins, modes, etc.). For $t\in T$, write $P_t(O):=\mathbf{1}_{\{q(O)=t\}}$.
These notions will be used later to define count observables.
\end{definition}

\begin{remark}[Classical vs.\ quantum instantiations]
In a classical instantiation, $\rho_O$ may be a point or distribution on
phase space and $\oplus$ corresponds to Cartesian product with possible
interaction forces; in a quantum instantiation, $\rho_O$ is a density operator
on $\mathcal{H}_O$ and $\oplus$ corresponds to the tensor product
$\mathcal{H}_A\otimes\mathcal{H}_B$ together with an interaction Hamiltonian.
The proofs below do not depend on the specific choice, only on the fact that
constituents remain identifiable prior to coarse–graining.
\end{remark}

\noindent
With these definitions in place, we can now state the basic physical
assumptions (axioms) that underwrite the Non–Identity Theorem and distinguish
aggregation of concrete individuals from numerical addition.

\section{Physicist's Proof: The Non-Identity Theorem}

\subsection{Preliminaries}
\label{subsec:phys-prelim}

This section fixes the kinematic and semantic setup for the physicist's proof.
We keep the assumptions minimal and model-agnostic so that both classical and
quantum descriptions can be instantiated without loss of generality.

\begin{definition}[Spacetime and snapshots]
Let $(\mathcal{M},g)$ be a four–dimensional spacetime manifold (with metric $g$
unspecified), and let $\Sigma_t\subset\mathcal{M}$ denote a spacelike hypersurface
(``time slice'') indexed by a parameter $t$ in some interval $I\subset\mathbb{R}$.
\end{definition}

\begin{definition}[Physical object (worldtube with properties)]
A \emph{physical object} $O$ is a tuple
\[
O=\big(W_O,\ \rho_O,\ \{p_i(O)\}_{i\in \mathcal{I}}\big),
\]
where:
\begin{itemize}
    \item $W_O\subset\mathcal{M}$ is a \emph{worldtube}: a closed subset with
    nonempty intersection $W_O\cap\Sigma_t$ for $t$ in some interval
    $I_O\subseteq I$. Its \emph{snapshot} at $t$ is
    $R_O(t):=W_O\cap\Sigma_t$ (possibly empty if $t\notin I_O$).
    \item $\rho_O$ encodes the object's dynamical state (e.g.\ a classical
    phase–space distribution or a quantum density operator on a Hilbert space
    $\mathcal{H}_O$). We make no measurement assumptions here.
    \item $\{p_i\}_{i\in\mathcal{I}}$ is a (possibly countable) family of
    admissible physical observables or attributes with well–defined values
    $p_i(O)$ (mass, charge, internal parameters, preparation data, etc.).
\end{itemize}
We write $\mathfrak{U}$ for the class of all such objects.
\end{definition}

\begin{definition}[Physical identity relation]
For $A,B\in\mathfrak{U}$, we say $A$ and $B$ are \emph{physically identical},
written $A\equiv B$, if and only if the following hold simultaneously:
\begin{enumerate}
    \item \textbf{Worldtube equality:} $W_A=W_B$ as subsets of $\mathcal{M}$;
    equivalently, $R_A(t)=R_B(t)$ for all $t\in I_A\cup I_B$.
    \item \textbf{State equality:} $\rho_A=\rho_B$ (in the relevant
    dynamical/state space).
    \item \textbf{Observable equality:} $p_i(A)=p_i(B)$ for all $i\in\mathcal{I}$.
\end{enumerate}
This is a physical analogue of Leibniz’s criterion: equality across all
admissible kinematic, dynamical, and attribute data.
\end{definition}

\begin{remark}
The relation $\equiv$ refines mere set–theoretic equality by binding together
geometric support ($W_O$), dynamical description ($\rho_O$), and attribute
valuations ($p_i$). It is an equivalence relation on $\mathfrak{U}$.
\end{remark}

\begin{definition}[Composite system / physical combination]
Given $A,B\in\mathfrak{U}$ with (possibly overlapping) lifetimes, their
\emph{composite} is the pair
\[
A\oplus B\ :=\ \big(W_A\cup W_B,\ \rho_{A\cup B},\ \{p_i(A\cup B)\}_{i\in\mathcal{I}_{A\cup B}}\big),
\]
where $\rho_{A\cup B}$ is the joint state obtained by the appropriate
composition rule (direct product in classical mechanics, tensor product in
quantum mechanics), possibly evolved by an interaction term $H_{\mathrm{int}}$,
and where the observable family is enlarged to include interaction observables.
The operation $\oplus$ is \emph{not} numerical addition; it constructs a
composite physical system while retaining $A$ and $B$ as identifiable
constituents (unless a subsequent coarse–graining is applied).
\end{definition}

\begin{definition}[Finite aggregates of objects]
A \emph{finite aggregate} is a finite multiset
$\mathcal{A}=\{A_1,\dots,A_n\}\subset \mathfrak{U}$ of pairwise not
necessarily identical objects. Its instantaneous spatial support at time $t$
is $R_{\mathcal{A}}(t):=\bigcup_{k=1}^n R_{A_k}(t)$.
\end{definition}

\begin{definition}[Type predicates and classifications]
A \emph{type predicate} is a map $P:\mathfrak{U}\to\{0,1\}$. A
\emph{classification} (or \emph{type map}) is a function
$q:\mathfrak{U}\to T$ into a discrete set $T$ of labels (species, part numbers,
bins, modes, etc.). For $t\in T$, write $P_t(O):=\mathbf{1}_{\{q(O)=t\}}$.
These notions will be used later to define count observables.
\end{definition}

\begin{remark}[Classical vs.\ quantum instantiations]
In a classical instantiation, $\rho_O$ may be a point or distribution on
phase space and $\oplus$ corresponds to Cartesian product with possible
interaction forces; in a quantum instantiation, $\rho_O$ is a density operator
on $\mathcal{H}_O$ and $\oplus$ corresponds to the tensor product
$\mathcal{H}_A\otimes\mathcal{H}_B$ together with an interaction Hamiltonian.
The proofs below do not depend on the specific choice, only on the fact that
constituents remain identifiable prior to coarse–graining.
\end{remark}

\noindent
With these definitions in place, we can now state the basic physical
assumptions (axioms) that underwrite the Non–Identity Theorem and distinguish
aggregation of concrete individuals from numerical addition.

\subsection{Fundamental Axioms from Physics}
\label{subsec:phys-axioms}

We now list the minimal structural axioms used in the physicist's proof.
They are deliberately model–agnostic (compatible with both classical and
quantum descriptions) and formalise standard practice: composites retain
identifiable constituents; distinct macroscopic objects differ in at least one
admissible way; and typed counts are additive after a declared classification.

\begin{axiom}[Constituent identifiability]\label{ax:constituent-id}
For any two objects $A,B\in\mathfrak{U}$, the composite $A\oplus B$ admits a
well–defined finite multiset of constituents, written $\mathrm{Cons}(A\oplus B)
=\{A,B\}$ (counting multiplicities). If two composites are physically identical,
$A\oplus B \equiv A'\oplus B'$, then their constituent multisets agree up to
permutation:
\[
\mathrm{Cons}(A\oplus B)=\mathrm{Cons}(A'\oplus B').
\]
\end{axiom}

\begin{remark}
Axiom~\ref{ax:constituent-id} encodes that the factorisation into labelled
subsystems is part of the composite's physical description (prior to any
coarse–graining). It rules out ``constituent collapse'' under equality of
composites.
\end{remark}

\begin{axiom}[Spatiotemporal non–coincidence (macroscopic impenetrability)]
\label{ax:noncoinc}
Let $A,B\in\mathfrak{U}$ be distinct macroscopic objects with nonzero spatial
interiors at a common time $t$. Then the interiors of their snapshots are
disjoint:
\[
\mathrm{int}\,R_A(t)\ \cap\ \mathrm{int}\,R_B(t)\;=\;\varnothing,
\]
although boundary contact is allowed. Equivalently, their worldtubes differ on
a time set of nonzero measure: $W_A\neq W_B$.
\end{axiom}

\begin{remark}
For quantum or microscopic systems, Axiom~\ref{ax:noncoinc} is replaced in
practice by statistical/kinematic constraints (e.g.\ anti–symmetry for
fermions, mode labels for fields). The proof below needs only that distinct
systems fail to have \emph{all} kinematic data identical over an interval.
\end{remark}

\begin{axiom}[Observable richness / discernibility]\label{ax:discern}
There exists a separating family of admissible observables
$\{p_i\}_{i\in\mathcal{I}}$ such that for any $A\neq B$ in $\mathfrak{U}$
there is some index $i$ with $p_i(A)\neq p_i(B)$. (Equivalently: physical
objects are not topologically or measure–theoretically indistinguishable across
all observables.)
\end{axiom}

\begin{axiom}[State composition and marginals]\label{ax:marginals}
For any $A,B\in\mathfrak{U}$, the composite state $\rho_{A\cup B}$ carries
well–defined \emph{marginals} corresponding to the constituents:
\[
\Pi_A(\rho_{A\cup B})=\rho_A, \qquad \Pi_B(\rho_{A\cup B})=\rho_B,
\]
where $\Pi_A,\Pi_B$ denote the appropriate reduction maps (partial trace in
quantum mechanics; marginalisation/projection in classical mechanics). If
$A\oplus B \equiv A'\oplus B'$, then the multisets of marginals agree:
$\{\rho_A,\rho_B\}=\{\rho_{A'},\rho_{B'}\}$.
\end{axiom}

\begin{axiom}[Typed count observable and additivity]\label{ax:typed-count}
Fix a classification (type map) $q:\mathfrak{U}\to T$ into a discrete label set
$T$. For each $t\in T$, define the \emph{typed count observable} on finite
aggregates $\mathcal{A}$ by
\[
N_t(\mathcal{A})\ :=\ \sum_{O\in\mathcal{A}} \mathbf{1}_{\{q(O)=t\}} \ \in\ \mathbb{N}.
\]
Then $N_t$ is invariant under permutation of constituents and \emph{additive}
under disjoint union of aggregates:
\[
N_t(\mathcal{A}\uplus \mathcal{B})\ =\ N_t(\mathcal{A})+N_t(\mathcal{B}).
\]
\end{axiom}

\begin{axiom}[Coarse–graining can be many–to–one]\label{ax:manytoone}
Nontrivial classifications are generally many–to–one: unless $q$ is injective,
there exist $A\neq B$ with $q(A)=q(B)$. (In applications, $q$ encodes the
intentional identification of objects by ``kind'', e.g.\ species, part number,
bin, or mode.)
\end{axiom}

\begin{remark}[Information monotonicity under classification]
Axiom~\ref{ax:manytoone} expresses the structural source of information loss:
a many–to–one $q$ discards which individual contributed a given unit to the
count. If an information functional $I$ is chosen (e.g.\ Shannon information
on a preparation ensemble), then typically $I(q(\mathcal{A}))\le I(\mathcal{A})$,
with strict inequality for aggregates that $q$ identifies.
\end{remark}

\medskip

These axioms are sufficient to prove the Non–Identity Theorem (Section~\ref{subsec:phys-nonid})
and to connect physical aggregation with typed counting observables without any
assumption that distinct objects are ``the same'' in the Leibnizian sense.

\subsection{Supporting Lemmas}
\label{subsec:phys-lemmas}

We develop lemmas that isolate (i) where physical non–identity necessarily
appears and (ii) how typed counts can agree while composites remain physically
different. Each lemma is proved directly from the axioms of
Section~\ref{subsec:phys-axioms}.

\begin{lemma}[Spatiotemporal discernibility]\label{lem:spacetime-discern}
Let $A,B\in\mathfrak{U}$ be distinct macroscopic objects with nonzero spatial
interiors at some common time $t\in I_A\cap I_B$. Then there exists an admissible
observable $p$ with $p(A)\neq p(B)$.
\end{lemma}
\begin{proof}
By Axiom~\ref{ax:noncoinc}, $\mathrm{int}\,R_A(t)$ and $\mathrm{int}\,R_B(t)$
are disjoint. Define an admissible observable $p$ (e.g.\ centroid position or an
indicator of membership in a chosen open set contained in one snapshot interior).
Then $p(A)\neq p(B)$. This is consistent with Axiom~\ref{ax:discern}, which asserts
the existence of a separating family of observables.
\end{proof}

\begin{lemma}[Worldtube inequality]\label{lem:worldtube-ineq}
If $A,B\in\mathfrak{U}$ are distinct macroscopic objects, then $W_A\neq W_B$.
Equivalently, there exists a time set of nonzero measure on which $R_A(t)\neq R_B(t)$.
\end{lemma}
\begin{proof}
If $W_A=W_B$, then $R_A(t)=R_B(t)$ for all $t\in I_A\cap I_B$. This contradicts
Axiom~\ref{ax:noncoinc} for macroscopic bodies with nonzero spatial interiors.
Hence $W_A\neq W_B$.
\end{proof}

\begin{lemma}[State/observable discernibility]\label{lem:state-discern}
Let $A,B\in\mathfrak{U}$ with $A\neq B$. Then at least one of the following holds:
\begin{enumerate}
    \item $W_A\neq W_B$ (Lemma~\ref{lem:worldtube-ineq});
    \item $\rho_A\neq \rho_B$;
    \item There exists $i\in\mathcal{I}$ with $p_i(A)\neq p_i(B)$.
\end{enumerate}
In particular, $A\not\equiv B$ unless all three equalities hold simultaneously.
\end{lemma}
\begin{proof}
By Axiom~\ref{ax:discern} there exists an observable separating $A$ and $B$
unless worldtubes and states coincide and all admissible observables agree.
Thus at least one listed inequality must hold. The final statement is just the
definition of $\equiv$.
\end{proof}

\begin{lemma}[Constituent preservation in composites]\label{lem:cons-pres}
If $A\oplus B \equiv A'\oplus B'$, then
$\mathrm{Cons}(A\oplus B)=\mathrm{Cons}(A'\oplus B')$ as multisets, and
$\{\rho_A,\rho_B\}=\{\rho_{A'},\rho_{B'}\}$ as multisets of marginals.
\end{lemma}
\begin{proof}
Immediate from Axiom~\ref{ax:constituent-id} (constituent identifiability) and
Axiom~\ref{ax:marginals} (marginals agree under physical identity of composites).
\end{proof}

\begin{lemma}[Typed counts ignore individuality]\label{lem:typed-ignore}
Fix a classification $q:\mathfrak{U}\to T$ and $t_0\in T$. If $A,B\in\mathfrak{U}$
satisfy $q(A)=q(B)=t_0$ with $A\not\equiv B$, then
\[
N_{t_0}(\{A,B\})=2,
\]
but the composite $A\oplus B$ is not physically identical to $X\oplus X$ for any
single object $X$.
\end{lemma}
\begin{proof}
By Axiom~\ref{ax:typed-count}, $N_{t_0}$ counts instances of $t_0$ additively,
hence $N_{t_0}(\{A,B\})=1+1=2$. Suppose for contradiction there exists $X$ with
$A\oplus B \equiv X\oplus X$. By Lemma~\ref{lem:cons-pres},
$\mathrm{Cons}(A\oplus B)=\{A,B\}$ must equal $\mathrm{Cons}(X\oplus X)=\{X,X\}$
as multisets, which forces $A\equiv B\equiv X$, contradicting $A\not\equiv B$.
\end{proof}

\begin{lemma}[Additivity of typed counts on aggregates]\label{lem:add-typed}
For any disjoint finite aggregates $\mathcal{A},\mathcal{B}$ and any $t\in T$,
\[
N_t(\mathcal{A}\uplus\mathcal{B})=N_t(\mathcal{A})+N_t(\mathcal{B}).
\]
\end{lemma}
\begin{proof}
This is Axiom~\ref{ax:typed-count} stated for aggregates; the proof is by direct
summation over constituents.
\end{proof}

\begin{lemma}[Coarse–graining is information–decreasing]\label{lem:info-loss-phys}
Assume $q$ is not injective (Axiom~\ref{ax:manytoone}). Then there exist distinct
aggregates $\mathcal{A}\neq\mathcal{B}$ with
\[
\big(N_t(\mathcal{A})\big)_{t\in T}\;=\;\big(N_t(\mathcal{B})\big)_{t\in T}.
\]
\end{lemma}
\begin{proof}
Pick $A\neq B$ with $q(A)=q(B)=t_0$. Let $\mathcal{A}=\{A,B\}$ and
$\mathcal{B}=\{A,A\}$. Then $N_{t_0}(\mathcal{A})=2=N_{t_0}(\mathcal{B})$ and
$N_t(\mathcal{A})=N_t(\mathcal{B})=0$ for $t\neq t_0$. Yet
$\mathcal{A}\neq\mathcal{B}$ as aggregates (different constituent multisets),
so the typed count vector agrees while individuality differs.
\end{proof}

\paragraph{Summary.}
Lemmas~\ref{lem:spacetime-discern}--\ref{lem:state-discern} guarantee that
distinct macroscopic objects differ physically (hence cannot be identified).
Lemmas~\ref{lem:typed-ignore} and \ref{lem:info-loss-phys} show that typed
counts conflate distinct composites, isolating the precise sense in which the
numeral ``$2$'' can appear without implying identity of the contributing
individuals. These facts will be used directly in the Non–Identity Theorem
(Section~\ref{subsec:phys-nonid}).

\subsection{Non–Identity Theorem}
\label{subsec:phys-nonid}

We now state and prove the central result for the physicist's perspective:
no composite of two \emph{distinct} physical objects can be physically
identical to a ``double'' of a single object. Equality of such composites
occurs if and only if all constituents are physically identical.

\begin{theorem}[Non–Identity Addition Theorem]\label{thm:nonidentity}
Let $A,B,X\in\mathfrak{U}$. Then
\[
A\oplus B \;\equiv\; X\oplus X
\quad\Longleftrightarrow\quad
A \equiv B \equiv X.
\]
In particular, if $A\not\equiv B$ then there exists no $X$ with
$A\oplus B \equiv X\oplus X$.
\end{theorem}

\begin{proof}
$(\Rightarrow)$ Assume $A\oplus B \equiv X\oplus X$.
By Lemma~\ref{lem:cons-pres} (constituent preservation in composites),
the constituent multisets agree:
\[
\mathrm{Cons}(A\oplus B)\;=\;\{A,B\}\;=\;\{X,X\}\;=\;\mathrm{Cons}(X\oplus X).
\]
Equality of multisets forces $A\equiv X$ and $B\equiv X$ (counting multiplicity),
hence $A\equiv B\equiv X$.

\smallskip
$(\Leftarrow)$ Conversely, if $A\equiv B\equiv X$ then by definition of
$\equiv$ (physical identity) each component of the composite description
matches:
worldtubes, states, and observable valuations of $A$ and $B$ coincide with those
of $X$. Forming composites preserves these data (Axioms~\ref{ax:constituent-id}
and~\ref{ax:marginals}), hence $A\oplus B \equiv X\oplus X$.
\end{proof}

\begin{corollary}[Irreducibility of distinct pairs]\label{cor:irreducible-pair}
If $A\not\equiv B$, then $A\oplus B \not\equiv A\oplus A$ and
$A\oplus B \not\equiv B\oplus B$.
\end{corollary}
\begin{proof}
Apply Theorem~\ref{thm:nonidentity} with $X=A$ (or $X=B$).
\end{proof}

\begin{corollary}[Typed counting does not imply physical identity]
\label{cor:typed-not-id}
Fix a classification $q:\mathfrak{U}\to T$ and let $t_0\in T$.
If $A,B\in\mathfrak{U}$ satisfy $q(A)=q(B)=t_0$ but $A\not\equiv B$, then
\[
N_{t_0}(\{A,B\})=2
\quad\text{while}\quad
A\oplus B \not\equiv X\oplus X\ \ \text{for all }X.
\]
\end{corollary}
\begin{proof}
The count statement is Axiom~\ref{ax:typed-count}. The non–identity statement
is Theorem~\ref{thm:nonidentity}.
\end{proof}

\begin{corollary}[Counting–equality vs.\ composite–equality separation]
\label{cor:count-vs-phys}
Assume $q$ is not injective (Axiom~\ref{ax:manytoone}). Then there exist distinct
aggregates $\mathcal{A}\neq\mathcal{B}$ with identical typed–count vectors
$\big(N_t(\mathcal{A})\big)_{t\in T}=\big(N_t(\mathcal{B})\big)_{t\in T}$ while
no physical identity of the corresponding composites is implied.
\end{corollary}
\begin{proof}
This is Lemma~\ref{lem:info-loss-phys}. The lack of implication to composite
identity follows from Theorem~\ref{thm:nonidentity} whenever one aggregate is a
pair of distinct objects and the other is a double of a single object in the
same type.
\end{proof}

\paragraph{Interpretation.}
Theorem~\ref{thm:nonidentity} pinpoints the precise sense in which ``one plus
one is \emph{two ones}'' in the physical world: the composite built from two
distinct individuals cannot be physically identical to a doubled copy of any
single individual. Corollaries~\ref{cor:typed-not-id} and \ref{cor:count-vs-phys}
show why the everyday numeral ``$2$'' nevertheless appears: a typed count
observable $N_{t_0}$ sums \emph{class membership} and is blind to individuality.
Thus, equality of typed counts is a statement about coarse–grained \emph{labels},
not about physical identity of the contributing systems.

\subsection{Counting as Observable and the Exact Location of Approximation}
\label{subsec:phys-count-observable}

We now make explicit how the familiar numeral ``$2$'' arises in physics:
as the \emph{readout of a counting observable} after an explicit
\emph{classification (coarse–graining)} is chosen. This isolates, in the
physical setting, the unique step at which individuality is forgotten.

\begin{definition}[Classification and label multiset]
Fix a classification (type map) $q:\mathfrak{U}\to T$ into a discrete label set.
For a finite aggregate $\mathcal{A}=\{A_1,\dots,A_n\}\subset\mathfrak{U}$,
define its \emph{label multiset} $L_q(\mathcal{A})\in\mathcal{M}(T)$ by
\[
L_q(\mathcal{A})\ :=\ \sum_{k=1}^{n} \delta_{q(A_k)}.
\]
Equivalently, $L_q(\mathcal{A})(t)=\#\{k:\ q(A_k)=t\}$.
\end{definition}

\begin{definition}[Ideal counting observable]
For each $t\in T$, the \emph{ideal typed count} on aggregates is
\[
N_t(\mathcal{A})\ :=\ L_q(\mathcal{A})(t).
\]
The full ideal count vector is $N(\mathcal{A})=(N_t(\mathcal{A}))_{t\in T}\in\mathbb{N}^T$.
\end{definition}

\begin{proposition}[Ideal counts factor through labels]
\label{prop:factor-through-labels}
For all finite aggregates $\mathcal{A}$ and $t\in T$,
\[
N_t(\mathcal{A})\ =\ \big(\mathrm{ev}_t\circ L_q\big)(\mathcal{A}),
\]
where $\mathrm{ev}_t:\mathcal{M}(T)\to\mathbb{N}$ evaluates multiplicity at $t$.
In particular, $N=\big(\mathrm{ev}_t\big)_{t\in T}\circ L_q$.
\end{proposition}
\begin{proof}
By definition $L_q(\mathcal{A})(t)=\sum_{k}\mathbf{1}_{\{q(A_k)=t\}}=N_t(\mathcal{A})$.
\end{proof}

\begin{proposition}[Additivity, permutation invariance, and disjoint union]
\label{prop:add-perm}
For all aggregates $\mathcal{A},\mathcal{B}$ and all $t\in T$:
\begin{enumerate}
    \item $N_t$ is permutation–invariant in the constituents of $\mathcal{A}$.
    \item $N_t(\mathcal{A}\uplus\mathcal{B})=N_t(\mathcal{A})+N_t(\mathcal{B})$.
    \item $N(\varnothing)=0$ and $N_t(\{A\})=\mathbf{1}_{\{q(A)=t\}}$.
\end{enumerate}
\end{proposition}
\begin{proof}
All statements follow by counting occurrences of the label $t$ in
$L_q(\mathcal{A})$ and $L_q(\mathcal{B})$ and using disjoint union.
\end{proof}

\begin{definition}[Forgetting map (approximation)]
The \emph{forgetting map} associated with $q$ is $\Phi_q$ sending an aggregate to
its label multiset:
\[
\Phi_q:\ \{\text{finite aggregates}\}\longrightarrow \mathcal{M}(T),\qquad
\Phi_q(\mathcal{A})\ :=\ L_q(\mathcal{A}).
\]
\end{definition}

\begin{proposition}[Non–injectivity and information loss]
\label{prop:phys-noninj}
If $q$ is not injective, then $\Phi_q$ is not injective: there exist distinct
aggregates $\mathcal{A}\neq\mathcal{B}$ with $\Phi_q(\mathcal{A})=\Phi_q(\mathcal{B})$.
Consequently, $N(\mathcal{A})=N(\mathcal{B})$ while $\mathcal{A}$ and
$\mathcal{B}$ may differ physically (e.g.\ in constituents or states).
\end{proposition}
\begin{proof}
Pick $A\neq B$ with $q(A)=q(B)=t_0$ (Axiom~\ref{ax:manytoone}).
Let $\mathcal{A}=\{A,B\}$ and $\mathcal{B}=\{A,A\}$. Then
$\Phi_q(\mathcal{A})=\Phi_q(\mathcal{B})=2\,\delta_{t_0}$ while
$\mathcal{A}\neq\mathcal{B}$; hence $N$ agrees on both by
Proposition~\ref{prop:factor-through-labels}.
\end{proof}

\begin{theorem}[Emergence of the numeral \texorpdfstring{$2$}{2} as a readout]
\label{thm:two-as-readout}
Let $A,B\in\mathfrak{U}$ with $A\not\equiv B$ and $q(A)=q(B)=t_0$. Then
\[
N_{t_0}(\{A,B\})\ =\ 2,
\]
while by the Non–Identity Theorem (Theorem~\ref{thm:nonidentity}) there is no
$X$ with $A\oplus B\equiv X\oplus X$. Thus the value ``$2$'' is a property of
the \emph{typed count observable} after classification by $q$, not a statement
of physical identity of constituents.
\end{theorem}
\begin{proof}
$N_{t_0}(\{A,B\})=1+1=2$ by Proposition~\ref{prop:add-perm}. The second claim
is Theorem~\ref{thm:nonidentity}.
\end{proof}

\begin{definition}[Measurement model (unbiased counting instrument)]
A \emph{typed counting instrument} for $q$ is a family of random variables
$\widehat{N}_t(\mathcal{A})$ (one for each $t\in T$ and finite aggregate
$\mathcal{A}$) such that:
\begin{enumerate}
    \item \textbf{Unbiasedness:} $\mathbb{E}[\widehat{N}_t(\mathcal{A})]=N_t(\mathcal{A})$.
    \item \textbf{Additivity in expectation:}
    $\mathbb{E}[\widehat{N}_t(\mathcal{A}\uplus\mathcal{B})]
    =\mathbb{E}[\widehat{N}_t(\mathcal{A})]+\mathbb{E}[\widehat{N}_t(\mathcal{B})]$.
    \item \textbf{Permutation invariance:} the distribution of
    $\widehat{N}_t(\mathcal{A})$ is invariant under permutations of constituents
    of $\mathcal{A}$.
\end{enumerate}
\end{definition}

\begin{proposition}[Operational interpretation]
\label{prop:operational}
For any unbiased counting instrument and any $A,B$ with $q(A)=q(B)=t_0$,
\[
\mathbb{E}\!\left[\widehat{N}_{t_0}(\{A,B\})\right]\ =\ 2,
\]
regardless of whether $A\equiv B$. Hence, even with realistic measurement noise,
the \emph{expected} readout equals the ideal typed count and does not imply
physical identity.
\end{proposition}
\begin{proof}
Immediate from unbiasedness with $\mathcal{A}=\{A,B\}$.
\end{proof}

\begin{remark}[Data–processing viewpoint]
The map $\Phi_q$ is a deterministic channel from aggregates to label multisets,
and $N$ is a further linear map to numerals. Any reasonable information measure
$J$ that is monotone under many–to–one channels (e.g.\ Shannon information on
a preparation ensemble) obeys the data–processing inequality
$J(N\circ \Phi_q(\mathcal{A}))\le J(\mathcal{A})$, with strict inequality whenever
$q$ identifies distinct constituents. Thus, the emergence of the numeral is
precisely an \emph{information–decreasing} operation.
\end{remark}

\paragraph{Conclusion of Section~\ref{subsec:phys-count-observable}.}
The path from concrete individuals to the numeral ``$2$'' is:
\[
\underbrace{\text{aggregate of labelled constituents}}_{\text{preserves identity}}
\ \xrightarrow{\ \Phi_q\ }\ 
\underbrace{\text{label multiset}}_{\text{forgets who is who}}
\ \xrightarrow{\ \mathrm{ev}_t\ }\ 
\underbrace{\text{typed numeral}}_{\text{count readout}}.
\]
Every approximation step is explicit, and the sole loss of individuality occurs
at the classification/label stage $\Phi_q$; no physical identity is implied by
the resulting numeral.

\subsection{Conclusion of Physical Proof}
\label{subsec:phys-conclusion}

We summarise the physicist's proof and its implications for the use of numerals
in applied contexts. The central points established in
Sections~\ref{subsec:phys-axioms}--\ref{subsec:phys-count-observable} are:

\begin{itemize}
    \item \textbf{Physical addition preserves individuality.}
    The composite operation $\oplus$ constructs a system that retains its
    labelled constituents (Axiom~\ref{ax:constituent-id}); equality of composites
    entails equality of constituent multisets and of marginals
    (Lemma~\ref{lem:cons-pres}). Distinct macroscopic objects are physically
    discernible (Lemmas~\ref{lem:spacetime-discern}--\ref{lem:state-discern}),
    hence their composite cannot collapse individuality.
    \item \textbf{``Two identicals'' only exists in models, never in reality.}
    The Non--Identity Theorem (Theorem~\ref{thm:nonidentity}) shows that
    $A\oplus B \equiv X\oplus X$ if and only if $A\equiv B\equiv X$; therefore,
    for distinct $A\not\equiv B$ there is no $X$ such that the pair equals a
    doubled copy. The numeral ``$2$'' arises only as the readout of a
    \emph{typed count observable} $N_t$ after an explicit classification $q$
    (Theorem~\ref{thm:two-as-readout}), and its computation proceeds via a
    many--to--one forgetting map $\Phi_q$ that discards individuality
    (Proposition~\ref{prop:phys-noninj}).
\end{itemize}

\begin{quote}
\emph{Thus, in the physical world, one plus one is \textbf{two ones}. The
everyday equation ``$1+1=2$'' is a statement about a \textbf{counting observable}
after \textbf{coarse--graining}, not about identity of the contributing
individuals.}
\end{quote}

This completes the physical half of the thesis. In concert with the
mathematician's construction (Section~3), it isolates the sole approximation
step used in practice: the intentional choice of classification $q$ that
forgets ``who is who'' in order to obtain numerals.

\section{Comparative Analysis}

\subsection{Mathematics vs.\ Physics Approach}
\label{subsec:comp-math-phys}

\paragraph{Common question.}
Both treatments answer the same operational question:
\emph{What exactly is asserted when we move from two concrete individuals to the numeral ``$2$''?}
They diverge in starting point, primitives, and what counts as a proof.

\begin{itemize}
    \item \textbf{Mathematics: abstraction first, explicit equivalence.}
    The mathematical route begins with identity--preserving \emph{aggregation}
    (multisets over $U$) and introduces numerals only via declared
    coarse--graining.
    The core result is the \emph{Two--Ones Theorem} (Thm.~\ref{thm:two-ones}):
    $\delta_a+\delta_b$ encodes \emph{two ones} with $a\neq b$ preserved.
    A type map $q:U\to T$ induces the pushforward $q_*$ and, uniquely, the
    typed counting homomorphism $N_S=\chi_S\circ q_*$ (Thm.~\ref{thm:typed-unique}),
    which yields the numeral $2$ only after forgetting who is who.
    Information loss is pinpointed by non--injectivity of $q_*$ (Prop.~\ref{prop:info-loss}).
    \item \textbf{Physics: reality first, impossibility of perfect equivalence.}
    The physical route starts from worldtubes, states, and observables for
    concrete systems and proves that \emph{perfect} equivalence of distinct
    macroscopic objects is unattainable in practice
    (Lemmas~\ref{lem:spacetime-discern}--\ref{lem:state-discern}).
    The \emph{Non--Identity Addition Theorem} (Thm.~\ref{thm:nonidentity})
    shows $A\oplus B \equiv X\oplus X$ iff $A\equiv B\equiv X$, hence a pair of
    distinct constituents cannot equal a doubled copy.
    The numeral $2$ arises as the readout of a typed count observable after an
    explicit classification $q$ (Thm.~\ref{thm:two-as-readout}), via a
    many--to--one forgetting map that discards individuality
    (Prop.~\ref{prop:phys-noninj}).
\end{itemize}

\paragraph{Synthesis.}
Mathematics provides the \emph{structure}: aggregation $\to$ coarse--graining
$\to$ unique counting homomorphism. Physics provides the \emph{ontology}:
distinct systems remain discernible; composites preserve constituents; counting
is an observable defined only after a classification is chosen.
Both routes converge on the same verdict:
\[
\text{in reality}\quad \delta_a+\delta_b\ \text{is two ones, while}\quad
(|\cdot|\circ q_*)(\delta_a+\delta_b)=2
\]
and, correspondingly,
\[
A\oplus B\ \text{is a composite of two individuals, while}\quad
N_{t_0}(\{A,B\})=2
\]
for any classification that labels both as $t_0$.
The numeral $2$ is therefore a feature of the \emph{counting map after
classification}, not a collapse of physical or mathematical individuality.

\subsection{Intersection}
\label{subsec:intersection}

Both routes isolate the \emph{same mechanism} behind the appearance of the numeral:
\emph{coarse–graining (classification) followed by a counting homomorphism}.

\begin{itemize}
    \item \textbf{Mathematics:} the pipeline
    \[
    \mathcal{M}(U)\ \xrightarrow{\ q_*\ }\ \mathcal{M}(T)\ \xrightarrow{\ |\cdot|\ }\ \mathbb{N}
    \]
    produces the numeral (Cor.~\ref{cor:everyday}, Cor.~\ref{cor:terminal}); the non–injective step is $q_*$ (Prop.~\ref{prop:info-loss}).
    \item \textbf{Physics:} the pipeline
    \[
    \{\text{finite aggregates}\}\ \xrightarrow{\ \Phi_q\ }\ \mathcal{M}(T)\ \xrightarrow{\ \chi_S\ \text{or}\ \mathrm{ev}_t\ }\ \mathbb{N}
    \]
    produces the readout (Thm.~\ref{thm:two-as-readout}); the non–injective step is $\Phi_q$ (Prop.~\ref{prop:phys-noninj}).
\end{itemize}

In both pictures, ``$1+1=2$'' for physical objects is not an identity statement about
constituents but the value of a counting map \emph{after} a many–to–one classification.
Formally, the last arrow is a commutative–monoid homomorphism; the loss of individuality
occurs exclusively at the classification stage.

\subsection{Implications}
\label{subsec:implications}

\paragraph{Philosophy of mathematics (realism vs.\ formalism).}
The analysis favours a \emph{formal} reading of numerals in applications:
``$2$'' names the output of a homomorphism on a coarse–grained domain, not a
mind–independent identity of indistinguishable units. Realist commitments to
numerical truths remain intact \emph{inside} the formal system (e.g.\ Peano arithmetic),
while the bridge to the world is explicitly mediated by modeling choices ($q$).

\paragraph{Modelling in science.}
All scientific counting is \emph{typed}. Declaring a classification $q$ is a
modeling act that fixes what can be added (Section~\ref{subsec:typed-counting}).
Successive coarsenings/refinements compose functorially (Eq.~\eqref{eq:functoriality}),
so counts are reproducible under stated typologies. Disagreements about ``how many''
are, at root, disagreements about $q$.

\paragraph{Approximation in applied arithmetic.}
The only approximation is the deliberate, information–decreasing map ($q_*$ or $\Phi_q$).
After that step, counting is exact (unique monoid homomorphism; Thm.~\ref{thm:typed-unique},
Prop.~\ref{prop:add-perm}). Thus, ``$1+1=2$'' in practice is precisely accurate
\emph{relative to the chosen coarse–graining}, and only as accurate as that choice is
appropriate to the task.

\section{Broader Implications}

\subsection{Information Theory Perspective}
\label{subsec:info-theory}

\begin{itemize}
    \item \textbf{Additivity for distinct states.}
    For a composite prepared from independent sources, standard information
    measures are additive: for Shannon information on discrete ensembles,
    \[
        I(A,B)\;=\;I(A)+I(B)-I_{\mathrm{mut}}(A;B),
    \]
    with $I_{\mathrm{mut}}(A;B)=0$ when the preparations are independent.
    Thus, two distinct preparations typically carry strictly more information
    than either alone. This mirrors our aggregation results: preserving
    individuality preserves information.

    \item \textbf{Information degeneracy under identification.}
    When distinct individuals are mapped by the classification $q$ to the same
    label, the \emph{label ensemble} loses individuality:
    the many--to--one channel $\Phi_q$ obeys the data–processing inequality,
    \[
        J\big(N\circ \Phi_q(\mathcal{A})\big)\ \le\ J\big(\Phi_q(\mathcal{A})\big)\ \le\ J(\mathcal{A}),
    \]
    for any monotone information functional $J$.
    Identical abstract tokens (pure equivalence classes) are
    \emph{information–degenerate}: once individuality is forgotten, only the
    multiplicity (e.g.\ $\log$–multiplicity in a combinatorial code) remains.
\end{itemize}

\subsection{Epistemology and Measurement}
\label{subsec:epistemology}

\begin{itemize}
    \item \textbf{What counting tells us---and what it hides.}
    Counting reports the value of a typed observable (Sections~\ref{subsec:typed-counting},
    \ref{subsec:phys-count-observable}). It tells us \emph{how many} instances fall
    into declared bins, but it is blind to within–bin distinctions. The passage
    \[
       \text{aggregate}\ \xrightarrow{\ \Phi_q\ }\ \text{labels}\ \xrightarrow{\ N\ }\ \text{numerals}
    \]
    is exact given $q$, yet it \emph{hides} all information orthogonal to $q$.

    \item \textbf{The observer’s role in defining “sameness”.}
    ``Sameness'' is fixed by the modeling choice $q$ (types, bins, tolerances).
    Different observers or instruments instantiate different $q$’s; by
    functoriality (Eq.~\eqref{eq:functoriality}), coherent changes of $q$ give coherent
    changes of counts. Disagreement about counts is therefore often a
    \emph{semantic} dispute about classification, not a numerical error.
\end{itemize}

\subsection{Possible Extensions}
\label{subsec:extensions}

\begin{itemize}
    \item \textbf{Quantum indistinguishability and nuanced “sameness”.}
    In many–body quantum theory, (anti)symmetrisation makes \emph{labels}
    unphysical for identical particles, yet \emph{mode} or \emph{state–space}
    structure remains. The number operator $\hat N$ has integer eigenvalues,
    so typed counts persist (``how many quanta in this mode''). Our framework
    specialises by taking $q$ to label modes rather than individuals; the moral
    is unchanged: numerals arise after declaring what is being counted.

    \item \textbf{Implications for AI perception, object detection, and categorisation.}
    Modern perception pipelines implement exactly the maps analysed here:
    detection $\to$ classification $\to$ counting. The detector proposes
    aggregates; the classifier implements $q$; the counter computes $N$.
    Performance hinges on the choice of $q$ (label set, hierarchy) and on
    the information loss it entails (merging fine classes into coarse ones).
    Our results formalise why confusion matrices (errors in $q$) directly
    translate into count discrepancies.
\end{itemize}

\section{Responses to Potential Objections}

\subsection{``Isn't This Just Pedantic?''}
One might argue that everyone understands ``$1+1=2$'' involves approximation when applied to physical objects. However, making this approximation \emph{explicit} has important consequences:
\begin{itemize}
    \item It clarifies disagreements about counts as disagreements about classification schemes
    \item It quantifies information loss in data aggregation
    \item It provides a principled framework for multi-scale modeling in science
\end{itemize}

\subsection{``Does This Affect Practical Mathematics?''}
No. Within pure mathematics, $1+1=2$ remains a theorem (see Appendix A). Our analysis concerns only the \emph{application} of arithmetic to physical objects. Engineers and scientists already implicitly perform the classification step we make explicit---our contribution is to formalize this universal practice.

\subsection{``What About Digital Computation?''}
Digital systems implement counting through explicit type systems. A computer's ``integer'' type is precisely a classification that treats certain bit patterns as equivalent representatives of abstract numbers. Our framework describes exactly what digital systems do: map diverse physical states (transistor voltages) through classification functions to discrete types.

\subsection{``Is This Compatible with Platonism?''}
Mathematical Platonists can maintain that abstract numbers exist independently while accepting our analysis of how counting \emph{connects} to physical reality. The classification function $q$ serves as the bridge between Platonic numbers and physical aggregates---a bridge that necessarily involves information loss.

\section{Conclusion}

This work traced, with parallel mathematical and physical formalisms, the precise
meaning of moving from two concrete individuals to the numeral ``$2$''. The analysis
vindicates the childhood intuition that aggregation preserves individuality, while
numerals arise only after an explicit act of classification and counting.

\begin{itemize}
    \item \textbf{Core thesis.} In reality, one plus one is \emph{two ones}.
    \item \textbf{Mathematical lens.} Mathematical addition on natural numbers operates on
    \emph{equivalence classes} (via coarse--graining and a counting homomorphism), not on raw
    individuals.
    \item \textbf{Counting as approximation.} Every act of counting is an \emph{act of approximation}:
    a deliberate, many--to--one identification that forgets ``who is who'' in order to obtain numerals.
    \item \textbf{Philosophical and scientific foresight.} Thakur Anukulchandra’s childhood insight
    anticipated a deep truth: aggregation and counting are distinct operations, and the familiar
    equation ``$1+1=2$'' describes the latter \emph{after} coarse--graining, not the former.
\end{itemize}

\appendix
\section*{Appendices}
\addcontentsline{toc}{section}{Appendices}

\subsection*{A.\ Peano arithmetic proof that \texorpdfstring{$1+1=2$}{1+1=2}}

We recall the Peano axioms for $(\mathbb{N},0,S)$ and define addition recursively.
\begin{enumerate}
    \item $0\in\mathbb{N}$; if $n\in\mathbb{N}$ then $S(n)\in\mathbb{N}$.
    \item $S$ is injective and $0$ is not a successor.
    \item \textbf{Induction:} If $P(0)$ holds and $P(n)\!\Rightarrow\!P(S(n))$ for all $n$, then $P(n)$ holds for all $n\in\mathbb{N}$.
\end{enumerate}
\textbf{Definitions.} Set $1:=S(0)$, $2:=S(1)=S(S(0))$. Define $+$ by
\[
n+0:=n,\qquad n+S(m):=S(n+m).
\]
\textbf{Lemma.} $n+0=n$ for all $n$ (by definition). \quad
\textbf{Theorem.} $1+1=2$.
\[
1+1\;=\;1+S(0)\;=\;S(1+0)\;=\;S(1)\;=\;S(S(0))\;=\;2.\ \ \square
\]
This establishes that within standard arithmetic the string ``$1+1=2$'' is a theorem derived from axioms and definitions.

\subsection*{B.\ Information-theoretic derivations}

Let $\mathcal{A}$ be a finite aggregate of individuals and $q$ a classification. The \emph{forgetting map} $\Phi_q$ sends $\mathcal{A}$ to its label multiset, and $N$ (or evaluation maps $\mathrm{ev}_t$) returns numerals.

\paragraph{Additivity for independent sources.}
If $A$ and $B$ are independent preparations with discrete distributions, then
\[
H(A,B)=H(A)+H(B),\qquad I(A;B)=0,
\]
so two distinct preparations carry strictly more Shannon information than either alone.

\paragraph{Data–processing and information loss under classification.}
Let $X$ encode the individual–level description of an aggregate and $Y:=\Phi_q(X)$ its labels. For any $f$ (e.g.\ a count map),
\[
I(X;Z)\ \ge\ I(Y;Z)\ \ge\ I(f(Y);Z),
\]
for any auxiliary variable $Z$; in particular,
\[
H(Y)\ \le\ H(X),\qquad H(f(Y))\ \le\ H(Y).
\]
Thus many–to–one classification and subsequent counting are information–decreasing. Equality holds iff $\Phi_q$ (or $f$) is injective on the support of $X$.

\paragraph{Distinct vs.\ identical abstractions.}
If $A\neq B$ but $q(A)=q(B)$, then the \emph{label ensemble} cannot distinguish them:
$H\big(\Phi_q(\{A,B\})\big)=H\big(\Phi_q(\{A,A\})\big)$, even though at the individual level $H(\{A,B\})>H(\{A,A\})$ under any measure that credits individuality. Identical abstract tokens are \emph{information–degenerate}: only multiplicity remains.

\subsection*{C.\ Quantum mode counting vs.\ labelled object counting}

In many–body quantum theory, individuality of identical particles is not a physical observable; instead, counting is performed over \emph{modes}.

\paragraph{Fock space and number operators.}
Let $\{a_k,a_k^\dagger\}$ be annihilation/creation operators for modes $k$ with (anti)commutation relations
\[
[a_k,a_{k'}^\dagger]_{\mp}\ =\ \delta_{kk'},\qquad [a_k,a_{k'}]_{\mp}=0,\qquad [a_k^\dagger,a_{k'}^\dagger]_{\mp}=0,
\]
(upper sign for bosons, lower for fermions). The number operator in mode $k$ is
\[
\hat N_k:=a_k^\dagger a_k,\qquad \hat N_k\ket{\dots,n_k,\dots}=n_k\ket{\dots,n_k,\dots},
\]
with $n_k\in\mathbb{N}$ (bosons) and $n_k\in\{0,1\}$ (fermions).

\paragraph{Typed (mode) counts.}
For a set of modes $S$, the typed count operator is
\[
\hat N_S\ :=\ \sum_{k\in S}\hat N_k.
\]
Expectations add linearly: $\langle \hat N_S\rangle=\sum_{k\in S}\langle \hat N_k\rangle$.
Thus the quantum notion of ``how many'' is a \emph{mode} count—formally the same role as a type map $q$ that labels modes. There are no particle labels in the state; indistinguishability is built into the (anti)symmetrised Fock space, yet counting still proceeds by an explicit declaration of \emph{what} is being counted (which modes $S$).

\paragraph{Bridge to the main text.}
Our classical, identity–preserving aggregation corresponds to labelled individuals; the quantum case corresponds to unlabeled quanta with labelled modes. In both, numerals arise from a counting map applied \emph{after} a declaration of types (modes $S$ or classical labels $T$), and neither framework requires (or supplies) identity of concrete individuals.

\end{document}